\documentclass[oneside,english]{amsart}
\usepackage[T1]{fontenc}
\usepackage[latin9]{inputenc}
\usepackage{varioref}
\usepackage{amsthm}
\usepackage{amssymb}

\makeatletter
\numberwithin{equation}{section} 
\numberwithin{figure}{section} 
\theoremstyle{plain}
\theoremstyle{plain}
\newtheorem{thm}{Theorem}
  \theoremstyle{plain}
  \newtheorem{lem}[thm]{Lemma}
  \theoremstyle{plain}
  \newtheorem{prop}[thm]{Proposition}
  \theoremstyle{remark}
  \newtheorem*{note*}{Note}
  \theoremstyle{remark}
  \newtheorem*{conclusion*}{Conclusion}
  \theoremstyle{remark}
  \newtheorem{note}[thm]{Note}
 \theoremstyle{definition}
  
  \theoremstyle{plain}
  
\theoremstyle{definition}
\newtheorem{defn}[thm]{Definition}
\newtheorem{rem}{Remark}

\usepackage{amsthm}
\usepackage{amsfonts}
\usepackage{amscd}
\usepackage{pb-diagram}

\newcommand{\dif}{\textrm{\textbf{d}}}

\newcommand{\cI}{{\mathcal I}}

\makeatother

\usepackage{babel}

\begin{document}

\title{AKS systems and Lepage equivalent problems}

\author{S. Capriotti}

\maketitle

\begin{abstract}
The integrable systems known as ``AKS systems'' admit a natural formulation in terms of a Hamiltonian picture. The Lagrangian side of these systems are far less known; a version in these terms can be found in \cite{feher-2002-301}. The purpose of these notes in to provide a novel description of AKS systems in terms of a variational system more general than usual in mechanics. Additionally, and using techniques borrowed from \cite{GotayCartan}, it was possible to build the Hamiltonian side of this variational problem, allowing us to establish the equivalence with the usual approach to these integrable systems.
\end{abstract}

\tableofcontents

\section{Introduction}

Any (regular) Lagrangian dynamical system has a Hamiltonian counterpart. The usual prescription for the construction of this alternate system relies heavily in the notion of Legrendre transformation. On the other side, the integrable systems known as AKS dynamical systems are described in the Hamiltonian side, and some interesting Lagrangian versions has been found recently \cite{feher-2002-301}. The following article describes a kind of Lagrangian setting for the so-called AKS systems. The main theoretical weapons used in the work are:
\begin{itemize}
\item The notion of \emph{Lepage equivalent variational problems} \cite{GotayCartan}.
\item An alternate description of mechanics known as \emph{non standard mechanics}, described in \cite{SantiJGM}.
\end{itemize}
We will set a non standard variational problem, inspired in the work of Feher \emph{et al.} cited above, and use the canonical Lepage equivalent of this variational problem in order to provide an ``almost Lagrangian'' version of AKS systems. 

\section{A brief summary of AKS dynamics}

In the following section we will define what we say when we say ``AKS systems''; the scheme presented here was adapted, with minor changes, from joint work of the author with H. Montani \cite{capriotti10:_dirac_lie}.

\subsection{Symmetries of a factorizable Lie group}
Let us begin by considering $G\times\mathfrak{g}^*$ as a $A\times B$-space, if as above $G\times\mathfrak{g}^*\simeq T^*G$ via left trivialization and we lift the $A\times B$-action on $G$ given by
\[
A\times B\times G\rightarrow G:\left(a,b;g\right)\mapsto agb^{-1}.
\]
By using the facts that the action is lifted and the symplectic form on $G\times\mathfrak{g}^*$ is exact, we can determine the momentum map associated to this action; then we obtain that
\begin{align}
J:&G\times\mathfrak{g}^*\rightarrow\mathfrak{b}^0\times\mathfrak{a}^0\cr
&\left(g,\sigma\right)\mapsto\left(\pi_{\mathfrak{b}^0}\left(\mbox{Ad}^\sharp_g\sigma\right),\pi_{\mathfrak{a}^0}\left(\sigma\right)\right)\label{ExprABMomentumMap}
\end{align}
where $\mbox{Ad}^{\sharp}$ indicates the coadjoint \emph{action} of $G$ on $\mathfrak{g}^*$. Let us now define the submanifold
\[
\Lambda_{\mu\nu}:=\left\{\left(g,\sigma\right)\in G\times\mathfrak{g}^*:\pi_{\mathfrak{b}^0}\left(\mbox{Ad}^\sharp_g\sigma\right)=\mu,\pi_{\mathfrak{a}^0}\left(\sigma\right)=\nu\right\}
\]
for each pair $\mu\in\mathfrak{b}^0,\nu\in\mathfrak{a}^0$. We have the following lemma.
\begin{lem}
Let $\sigma_\pm\in\mathfrak{g}_\mp^0,a\in A,b\in B$ be arbitrary elements. Then the formulas
\begin{align*}
a\cdot\sigma_+&:=\pi_{\mathfrak{b}^0}\left(\mbox{Ad}^\sharp_a\sigma_+\right)\cr
b\cdot\sigma_-&:=\pi_{\mathfrak{a}^0}\left(\mbox{Ad}^\sharp_b\sigma_-\right)
\end{align*}
defines an action of $A$ (resp. $B$) on $\mathfrak{a}^0$ (resp. $\mathfrak{b}^0$); in fact, under the identifications
\[
\mathfrak{a}^*\simeq\mathfrak{b}^0,\mathfrak{b}^*\simeq\mathfrak{a}^0
\]
induced by the decomposition $\mathfrak{g}=\mathfrak{a}\oplus\mathfrak{b}$ these actions are nothing but the coadjoint actions of each factor on the dual of its Lie algebras.
\end{lem}
\begin{note}
The symbols $\mathcal{O}^{A}_{\sigma_+}$ (resp. $\mathcal{O}^{B}_{\sigma_-}$) will denote the orbit in $\mathfrak{b}^0$ (resp. $\mathfrak{a}^0$) under the actions defined in the previous lemma. Additionally, for each $\xi\in\mathfrak{g}^*$, the form $\xi^\flat\in\mathfrak{g}^*$ is given by $\xi^\flat:={B}\left(\xi,\cdot\right)$, where ${B}\left(\cdot,\cdot\right)$ is the invariant bilinear form on $\mathfrak{g}^*$ by the Killing form.
\end{note}

\subsection{AKS systems as reduced spaces}
Therefore we will have that $\Lambda_{\mu\nu}=J^{-1}\left(\mu,\nu\right)$, and taking into account the M-W reduction (see \cite{A-M}) the projection of $\Lambda_{\mu\nu}$ on $\Lambda_{\mu\nu}/A_\mu\times B_\nu$ is presymplectic, and the solution curves for the dynamical system defined there by the invariant Hamiltonian $H\left(g,\sigma\right):=\frac{1}{2}\sigma\left(\sigma^\flat\right)$ are closely related with the solution curves of the system induced in the quotient. To work with these equations, let us introduce some convenient coordinates; the map
\begin{align*}
L_{\mu\nu}:\Lambda_{\mu\nu}\rightarrow\mathcal{O}^{A}_\mu\times\mathcal{O}^{B}_\nu:\left(g,\sigma\right)&\mapsto\left(\pi_{\mathfrak{b}^0}\left(\mbox{Ad}^\sharp_{g_+^{-1}}\mu\right),\pi_{\mathfrak{a}^0}\left(\mbox{Ad}^\sharp_{g_-}\nu\right)\right)\cr
&\mapsto\left(\pi_{\mathfrak{b}^0}\left(\mbox{Ad}^\sharp_{g_-}\sigma\right),\pi_{\mathfrak{a}^0}\left(\mbox{Ad}^\sharp_{g_-}\sigma\right)\right),
\end{align*}
where $g=g_+g_-$, induces a diffeomorphism on $\Lambda_{\mu\nu}/\left(A_\mu\times B_\nu\right)$. If $\left(g,\xi;\sigma,\eta\right)$ is a tangent vector to $G\times\mathfrak{g}^*$ (all the relevant bundles are left trivialized) then the derivative of $L_{\mu\nu}$ can be written as
\begin{multline}\label{DerivadaLmunu}
\left.\left(L_{\mu\nu}\right)_*\right|_{\left(g,\sigma\right)}\left(g,\sigma;\xi,\eta\right)=\left(-\pi_{\mathfrak{b}^0}\left(\mbox{ad}^\sharp_{\xi_+}\mbox{Ad}^\sharp_{g_+^{-1}}\mu\right),\pi_{\mathfrak{a}^0}\left(\mbox{ad}^\sharp_{\xi_-}\mbox{Ad}^\sharp_{g_-}\nu\right)\right)\cr
=\left(\pi_{\mathfrak{b}^0}\left(\mbox{ad}^\sharp_{\xi_-}\mbox{Ad}^\sharp_{g_-}\sigma\right)+\pi_{\mathfrak{b}^0}\left(\mbox{Ad}^\sharp_{g_-}\eta\right),\pi_{\mathfrak{a}^0}\left(\mbox{ad}^\sharp_{\xi_-}\mbox{Ad}^\sharp_{g_-}\sigma\right)+\pi_{\mathfrak{a}^0}\left(\mbox{Ad}^\sharp_{g_-}\eta\right)\right)
\end{multline}
if and only if $g=g_+g_-,\xi_+=\pi_{\mathfrak{a}}\left(\mbox{Ad}_{g_-}\xi\right),\xi_-=\pi_{\mathfrak{b}}\left(\mbox{Ad}_{g_-}\xi\right)$. So the following remarkable result is true.

\begin{prop}
Let $\mathcal{O}^{A}_\mu\times\mathcal{O}^{B}_\nu$ be the phase space whose symplectic structure is $\omega_{\mu\nu}=\omega_\mu-\omega_\nu$, where $\omega_{\mu,\nu}$ are the corresponding Kirillov-Kostant symplectic structures on each orbit. If $i_{\mu\nu}:\Lambda_{\mu\nu}\hookrightarrow G\times\mathfrak{g}^*$ is the inclusion map, then we will have that
\[
i_{\mu\nu}^*\omega=L_{\mu\nu}^*\omega_{\mu\nu}.
\]
\end{prop}

A proof of this proposition can be found in \cite{capriotti10:_dirac_lie}.

\subsection{The dynamical data for AKS systems}\label{DynamicalDataAKS}
Therefore $\mathcal{O}^{A}_\mu\times\mathcal{O}^{B}_\nu$ with the symplectic structure $\omega_{\mu\nu}$ is symplectomorphic to the reduced space associated to the $A\times B$-action defined above on $G\times\mathfrak{g}^*$. As was pointed out before, the hamiltonian $\mathsf{H}\left(g,\sigma\right)=\frac{1}{2}\sigma\left(\sigma^\flat\right)$ is invariant for this action, and this implies that the solutions of the dynamical system defined by such a hamiltonian on $G\times\mathfrak{g}^*$ are in one-to-one correspondence with those of the dynamical system induced on $\mathcal{O}^{A}_\mu\times\mathcal{O}^{B}_\nu$ by the hamiltonian $\mathsf{H}_{\mu\nu}$ \cite{A-M} defined through
\[
i_{\mu\nu}^*\mathsf{H}=L_{\mu\nu}^*\mathsf{H}_{\mu\nu}.
\]
Let us note now that if $L_{\mu\nu}\left(g,\sigma\right)=\left(\omega_1,\omega_2\right)$, then $\mbox{Ad}^\sharp_{g_-}\sigma=\omega_1+\omega_2$, and so
\begin{align*}
\mathsf{H}_{\mu\nu}\left(\omega_1,\omega_2\right)&=\frac{1}{2}\left(\mbox{Ad}_{g_-^{-1}}^\sharp\left(\omega_1+\omega_2\right)\right)\left[\left(\mbox{Ad}_{g_-^{-1}}^\sharp\left(\omega_1+\omega_2\right)\right)^\flat\right]\cr
&=\frac{1}{2}\left(\mbox{Ad}_{g_-^{-1}}^\sharp\left(\omega_1+\omega_2\right)\right)\left[\mbox{Ad}_{g_-^{-1}}\left(\omega_1+\omega_2\right)^\flat\right]\cr
&=\frac{1}{2}\omega_1\left(\omega_1^\flat\right)+\frac{1}{2}\omega_2\left(\omega_2^\flat\right)+\omega_1\left(\omega_2^\flat\right).
\end{align*}
Therefore
\begin{multline*}
\left.\dif \mathsf{H}_{\mu\nu}\right|_{\left(\omega_1,\omega_2\right)}\left(\pi_{\mathfrak{b}^0}\left(\mbox{ad}^\sharp_\xi\omega_1\right),\pi_{\mathfrak{a}^0}\left(\mbox{ad}^\sharp_\zeta\omega_2\right)\right)=\cr
=\pi_{\mathfrak{b}^0}\left(\mbox{ad}^\sharp_\xi\omega_1\right)\left(\omega_1^\flat+\omega_2^\flat\right)+\pi_{\mathfrak{a}^0}\left(\mbox{ad}^\sharp_\zeta\omega_2\right)\left(\omega_1^\flat+\omega_2^\flat\right)
\end{multline*}
and the hamiltonian vector field will be given by
\begin{equation}\label{CampoVectHmunu0}
\left.V_{\mathsf{H}_{\mu\nu}}\right|_{\left(\omega_1,\omega_2\right)}=\left(\pi_{\mathfrak{b}^0}\left(\mbox{ad}^\sharp_{\pi_{\mathfrak{a}}\left(\omega_1^\flat+\omega_2^\flat\right)}\omega_1\right),\pi_{\mathfrak{a}^0}\left(\mbox{ad}^\sharp_{\pi_{\mathfrak{b}}\left(\omega_1^\flat+\omega_2^\flat\right)}\omega_2\right)\right).
\end{equation}
In general the solution of dynamical systems via reduction has a definite direction, opposite to the adopted by us in this discussion: That is, it is expected that the reduced system is easiest to solve, because it involves less degrees of freedom, and the solution of the original system is found by lifting the solution of the reduced system. In case of AKS systems, we proceed in the reverse direction: To this end let us note that the dynamical system on $G\times\mathfrak{g}^*$ determined by $\mathsf{H}$ has hamiltonian vector field given by
\begin{equation}\label{CampoVectHmunuSimp}
\left.V_\mathsf{H}\right|_{\left(g,\sigma\right)}=\left(\sigma^\flat,-\mbox{ad}^\sharp_{\sigma^\flat}\sigma\right)=\left(\sigma^\flat,0\right)
\end{equation}
because of the invariance condition $\mbox{ad}^\sharp_\xi\xi^\flat=0$ for all $\xi\in\mathfrak{g}$. Then the solution in the unreduced space passing through $\left(g,\sigma\right)$ at the initial time is
\[
t\mapsto\left(g\exp{t\sigma^\flat},\sigma\right),
\]
and if this initial data verifies $\pi_{\mathfrak{b}^0}\left(\mbox{Ad}^\sharp_g\sigma\right)=\mu,\pi_{\mathfrak{a}^0}\left(\sigma\right)=\nu$, then this curve will belong to $\Lambda_{\mu\nu}$ for all $t$. Therefore the map
\[
t\mapsto\left(\pi_{\mathfrak{b}^0}\left(\mbox{Ad}^\sharp_{\left(g_+\left(t\right)\right)^{-1}}\mu\right),\pi_{\mathfrak{a}^0}\left(\mbox{Ad}^\sharp_{g_-\left(t\right)}\nu\right)\right),
\]
where $g_\pm:\mathbb{R}\rightarrow G_\pm$ are the curves defined by the factorization problem $g_+\left(t\right)g_-\left(t\right)=g\exp{t\sigma^\flat}$, is solution for the dynamical system associated to the vectorial field \eqref{CampoVectHmunu0}, which is more difficult to solve than the original system.

\section{Non standard variational problems}

We will consider that a \emph{variational problem} \cite{GotayCartan} is a triple $\left(\Lambda,\lambda,\cI\right)$, where $p:\Lambda\rightarrow M$ is a fiber bundle on an $n$-dimensional manifold $M$, $\lambda$ is an $n$-form on $\Lambda$, and $\cI\subset\Omega^\bullet\left(\Lambda\right)$ is an EDS there. A dynamical system can be attached to these data by selecting an intermediate bundle
\[
\Lambda\rightarrow\Lambda_1\rightarrow M;
\]
for example, the variational problem for classical mechanics can be obtained from the variational problem $\left(TQ\times I\rightarrow I,L\dif t,\left<\dif q^i-v^i\dif t\right>_{\text{diff}}\right)$ by picking up the intermediate bundle $TQ\times I\rightarrow Q\times I\rightarrow I$. The new data allow us to formulate the variational problem of classical mechanics in the following form: To find a section of $Q\times I\rightarrow I$ such that it is an extremal of the action
\[
S\left[\gamma\right]:=\int_I\dot{\gamma}^*\left(L\dif t\right),
\]
where $\dot{\gamma}:I\rightarrow TQ\times I$ is a section satisfying the folllowing requeriments
\begin{itemize}
\item It makes commutative the diagram
\[
\begin{diagram}
\node[2]{TQ\times I}\arrow{s,r}{\tau_Q\times\text{id}}\\
\node{I}\arrow{ne,t}{\dot{\gamma}}\arrow{e,b}{\gamma}\node{Q\times I}
\end{diagram}
\]
\item It is an integral section for $\left<\dif q^i-v^i\dif t\right>_{\text{diff}}$.
\end{itemize}
The introduction of an intermediate bundle in a variational problem $\left(\Lambda,\lambda,\cI\right)$ allow us to consider the EDS $\cI$ as a \emph{prolongation structure}, and thus providing conditions that associate to every section of this intermediate bundle a section of the full bundle.

\section{Lepage equivalence of variational problems}

Let us introduce a method to work with a variational problem of the general kind that we are considering here. The contents of the following section are adapted from the author's work \cite{SantiJGM}.

\subsection{(Multi)hamiltonian formalism through Lepagean equivalent problems}

We want to construct a hamiltonian version for the non standard variational
problem. The usual approach \cite{Gotay:1997eg} seems useless here, because of the following facts:
\begin{itemize}
\item The covariant multimomentum space is a bundle in some sense dual to the velocity space, which is a jet space. 
\item The dynamics in the multimomentum space is defined through the Legendre transform, and it is not easy to generalize to a non standard problem this notion.
\end{itemize}
The trick to circumvect the difficulties is to mimic the passage from Hamilton's principle to Hamilton-Pontryaguin principle. This is done by including the generators of the prolongation structure in the lagrangian density by means of a kind of Lagrange multipliers. This procedure will be formalized below, where the hamiltonian version is defined by associating a first order variational
problem to the non standard variational problem, whose extremals are in one to one correspondence with its extremals. This is called \emph{canonical bivariant Lepage equivalent problem}.

\subsection{Lepagean equivalent problems}
Here we will follow closely the exposition of the subject in the article \cite{GotayCartan}. Before going into details, let us introduce a bit of terminology: If $\Lambda\stackrel{\pi}{\longrightarrow}M$ is a bundle, $\lambda\in\Omega^n\left(\Lambda\right)$ ($n=\text{dim}\,M$) and $\cI$ is an EDS on $\Lambda$, the symbol $\left(\Lambda\stackrel{\pi}{\longrightarrow}M,\lambda,\cI\right)$ indicates the variational problem consisting in extremize the action
\[
S\left[\sigma\right]=\int_M\sigma^*\left(\lambda\right)
\]
with $\sigma\in\Gamma\left(\Lambda\right)$ restricted to the set of integral sections of $\cI$. Furthermore, $\mathcal{E}\left(\lambda\right)$ will denote the set of extremals for $\left(\Lambda\stackrel{\pi}{\longrightarrow}M,\lambda,\cI\right)$.

The idea is to eliminate in some way the constraints imposed by the elements of $\mathcal{I}$; intuitively, it is expected that the number of unknown increase when this is done. The following concept captures these ingredients formally.

\begin{defn}[Lepage equivalent variational problem]
A \emph{Lepagean equivalent} of a variational problem $\left(\Lambda\stackrel{\pi}{\longrightarrow}M,\lambda,\cI\right)$ is another variational problem
\[
\left(\tilde{\Lambda}\stackrel{\rho}{\longrightarrow}M,\tilde{\lambda},\left\{0\right\}\right)
\]
together with a surjective submersion $\nu:\tilde{\Lambda}\rightarrow\Lambda$ such that
\begin{itemize}
\item $\rho=\pi\circ\nu$, and
\item if $\gamma\in\Gamma\left(\tilde{\Lambda}\right)$ is such that $\nu\circ\gamma$ is an integral section of $\cI$, then
\[
\gamma^*\tilde{\lambda}=\left(\nu\circ\gamma\right)^*\lambda.
\]
\end{itemize}
\end{defn}
There exists a canonical way to build up a Lepage equivalent problem associated to a given variational problem $\left(\Lambda\stackrel{\pi}{\longrightarrow}M,\lambda,\cI\right)$, the so called canonical Lepage equivalent problem. Let $\cI$ be differentially generated by the sections of a graded subbundle $I\subset\bigwedge^\bullet\left(T^*\Lambda\right)$ (this is a ``constant rank'' hypothesis, ensuring the existence of a bundle in the construction, see below). Define $\cI^{\text{alg}}$ as the algebraic ideal in $\Omega^\bullet\left(\Lambda\right)$ generated by $\Gamma\left(I\right)$, and
\[
\left(\cI^{\text{alg}}\right)^l:=\cI^{\text{alg}}\cap\Omega^l\left(\Lambda\right).
\]
For $\lambda\in\Omega^n\left(\Lambda\right)$, define the affine subbundle $W^\lambda\subset\bigwedge^n\left(T^*\Lambda\right)$ whose fiber above $p\in\Lambda$ is
\[
\left.W^\lambda\right|_p:=\left\{\left.\lambda\right|_p+\left.\beta\right|_p:\beta\in\left(\cI^{\text{alg}}\right)^n\right\}.
\]
\begin{defn}[Canonical Lepage equivalent problem]\label{DefCanonicalLepageEquivalent}
In the previous setting, it is the triple $\left(W^\lambda\stackrel{\rho}{\longrightarrow}M,\tilde{\Theta},\left\{0\right\}\right)$, where $\nu$ is the canonical projection $\bar{\tau}^n_\Lambda:\bigwedge^n\left(T^*\Lambda\right)\rightarrow\Lambda$ restricted to $W^\lambda$, $\rho:=\pi\circ\nu$ and $\tilde{\Theta}$ is the pullback of the canonical $n$-form
\[
\left.\Theta_n\right|_\alpha:=\alpha\circ\left(\bar{\tau}^n_\Lambda\right)_*
\]
to $W^\lambda$. The form $\tilde\Theta$ will be called \emph{Cartan form} of the variational problem.
\end{defn}
\begin{rem}
It is worth remarking that our terminology was brought from \cite{GotayCartan}, which is slightly different from the classical theory as exposed in e.g. \cite{KrupkaLagrangeanStructures}. It is fully explained in the former work how to be relate both approaches.
\end{rem}

Returning to our main concern, it can be proved that the canonical Lepage equivalent is a Lepagean equivalent problem of $\left(\Lambda\stackrel{\pi}{\longrightarrow}M,\lambda,\cI\right)$. Now, the extremals of some variational problem has, in general, nothing to do with the extremals of its Lepagean equivalent problem, so it is necessary to introduce the following definition.
\begin{defn}[Covariant and contravariant Lepage equivalent problems]
We say that a Lepagean equivalent problem $\left(W^\lambda\stackrel{\rho}{\longrightarrow}M,\tilde{\Theta},\left\{0\right\}\right)$ for the variational problem $\left(\Lambda\stackrel{\pi}{\longrightarrow}M,\lambda,\cI\right)$ is \emph{covariant} if $\nu\circ\gamma\in\mathcal{E}\left(\lambda\right)$ for all $\gamma\in\mathcal{E}\left(\tilde{\Theta}\right)$; on the contrary, it is called \emph{contravariant} if every $\sigma\in\mathcal{E}\left(\lambda\right)$ is the projection of some extremal in $\mathcal{E}\left(\tilde{\Theta}\right)$ through $\nu$. A Lepagean equivalent problem is \emph{bivariant} if and only if it is both covariant and contravariant.
\end{defn}
There exists a fundamental relation between the extremals of a variational problem and the extremals of its canonical Lepage equivalent.
\begin{thm}
The canonical Lepage equivalent is covariant.
\end{thm}
For a proof, see \cite{GotayCartan}. The contravariant nature of a Lepage
equivalent problem is more subtle to deal with; in fact, it must be verified in each case separately. 

\subsection{Canonical Lepage equivalent of a non standard problem}
We will apply these considerations to our problem. The important thing to note is that, if a variational problem has a covariant and contravariant Lepagean equivalent problem, then the latter can be considered as a kind of Hamilton-Pontryaguin's principle for the given variational problem; in fact, it is shown below that the canonical Lepagean equivalent problem associated to the variational problem underlying the Hamilton's principle gives rise to the classical Hamilton-Pontryaguin's principle. This will be our starting point for assigning a multisymplectic space to the variational problem we are dealing with. The bivariance ensures us that every extremal has been taken into account in the new setting. Otherwise, namely, for non contravariant canonical Lepage equivalent problems, some extremals for the original variational problem could be lost in the process.

So let
us suppose that we have the non standard problem defined by the following
data\[
\begin{cases}
\Lambda\rightarrow\Lambda_{1}\rightarrow M,\\
\mathcal{I}\subset\Omega^{\bullet}\left(\Lambda\right),\\
S\left[\sigma\right]:=\int_{M}\left(\textsf{\textbf{pr}}\sigma\right)^{*}\left(\lambda\right).\end{cases}\]
If $\cI$ in this non standard problem has the required regularity (i.e. a ``constant rank'' hypothesis, see reference \cite{SantiJGM}), then the variational problem $\left(\Lambda\rightarrow M,\lambda,\cI\right)$ will have a canonical Lepage equivalent problem $\left(\tilde{W}^\lambda\rightarrow M,\tilde{\Theta},\left\{0\right\}\right)$, and we can apply the scheme described above. In order
to carry out this task locally, let us suppose as above that the fibers of the bundle $I$ on an open set $U\subset\Lambda$ can be written as
\[
\left.I\right|_\gamma=\mathbb{R}\left\langle\left.\alpha_{1}^{1}\right|_\gamma,\cdots,\left.\alpha_{k_{1}}^{1}\right|_\gamma,\left.\alpha_{1}^{2}\right|_\gamma,\cdots,\left.\alpha_{k_{2}}^{2}\right|_\gamma,\cdots,\left.\alpha_{1}^{p}\right|_\gamma,\cdots,\left.\alpha_{k_{p}}^{p}\right|_\gamma\right\rangle,\quad\gamma\in U
\]
so that the prolongation
structure $\cI$ is (differentially) generated by\[
\cI=\left\langle \alpha_{i}^{j}:1\leq i\leq p,1\leq j\leq k_i\right\rangle _{\text{diff}}\]
on $U$, where $\left\langle \alpha_{1}^{j},\cdots,\alpha_{k_{j}}^{j}\right\rangle _{\text{diff}}=\cI^{\left(j\right)}=:\cI\cap\Omega^{j}\left(U\right)$ and $\alpha_i^j\in Z_1\left(\Lambda\right)$ for all $i,j$.
Then if $U_0:=\pi\left(U\right)\subset M$, we can define for each $1\leq l\leq p$ the numbers $m_{l}:=\mbox{dim}\left(M\right)-l$
and the $n$-form on $\tilde{\Lambda}_U:=U\times_{U_0}\bigoplus_{l=1}^p\left[\bigwedge^{m_l}\left(T^*U_0\right)\right]^{\oplus k_l}$ will reads\[
\tilde{\lambda}:=\sum_{l=1}^{p}\left(\sum_{j=1}^{k_{l}}\alpha_{j}^{l}\wedge\beta^{j}_{m_l}\right)-\lambda\]
where $\left(\beta^{1}_{m_{1}},\cdots,\beta^{k_{1}}_{m_{1}},\cdots,\beta^{1}_{m_{p}},\cdots,\beta^{k_{p}}_{m_{p}}\right)$
denotes sections of $\bigwedge^{m_l}\left(T^*U_0\right)$; the subscript in these sections thus indicates their degree (in the exterior algebra sense). It can be shown in the
examples below that, in many important cases, the Euler-Lagrange
equations associated to the non standard problem defined by the data\[
\begin{cases}
\tilde{\Lambda}_U\stackrel{\mbox{id}}{\longrightarrow}\tilde{\Lambda}_U\rightarrow M,\\
0\subset\Omega^{\bullet}\left(\tilde\Lambda_U\right),\\
S\left[\sigma\right]:=\int_{M}\left(\textsf{\textbf{pr}}\sigma\right)^{*}\tilde{\lambda},\end{cases}\]
has a family of solutions which is isomorphic to the family of solutions
of the previous system; this means that in these cases the canonical
Lepage equivalent problem is also contravariant. The Euler-Lagrange eqs for an stationary section $\sigma\in\Gamma\left(\tilde{\Lambda}_U\right)$
are\begin{equation}
\sigma^{*}\left(V\lrcorner\dif\tilde{\lambda}\right)=0,\qquad\forall V\in\Gamma\left(V\tilde{\Lambda}_U\right)\label{eq:HamEqElectro1}\end{equation}
because there are no conditions for admisibility of variations; these
sections are then integral sections for the EDS\begin{equation}
\cI:=\left\langle V\lrcorner\dif\tilde{\lambda}:V\in\Gamma\left(V\tilde{\Lambda}_U\right)\right\rangle _{\text{diff}}.\label{eq:EDSHam}\end{equation}
We call this EDS the (local version of) \emph{Hamilton-Cartan EDS}.

\section{The Lagrangian side of AKS systems}

\subsection{Prolongation and constraints}
The purpose of this section is to use our description of dynamical systems as a way of thinking about AKS systems. In order to motivate the essential ideas, let us consider the free dynamical system on $\mathbb{R}$ described by the lagragian $I\times T\mathbb{R}\times\mathbb{R}$
\[
\left.L\right|_{\left(t;q,v;w\right)}:=\left[\frac{1}{2}v^2-v_0w\right]\dif t
\]
and the prolongation structure
\[
\left.\theta\right|_{\left(t;q,v;w\right)}:=\dif q-\left(v-w\right)\dif t.
\]
This means that the allowed variations must keep invariant the equation $\dif\delta q-\left(\delta v-\delta w\right)\dif t=0$; if we vary the curve $\gamma:t\mapsto\left(t;q\left(t\right),v\left(t\right);w\left(t\right)\right)$, the variation of the associated action will be (by neglecting boundary terms)
\begin{align*}
\delta S&=\int_I\gamma^*\left(v\delta v-v_0\delta w\right)\dif t\cr
&=\int_I\gamma^*\left[v\left(\dif\delta q+\delta w\dif t\right)-v_0\delta w\right]\cr
&=\int_I\gamma^*\left(\delta q\dif v+\left(v-v_0\right)\delta w\dif t\right).
\end{align*}
So the Euler-Lagrange equations for this dynamical system will be
\[
\begin{cases}
\dot{v}=0,&\cr
v=v_0.&
\end{cases}
\]
This oversimplificated example clearly shows a fundamental characteristic involved with this kind of modifications in the prolongation structure: The ability of introducing constraints on the velocities.
\newline
On the other hand, for $G$ a factorizable Lie group $G=AB$, and taking into account $G\times\mathfrak{g}^*$ with the canonical symplectic structure (using left trivialization). The assumed factorization induces the decompositions
\begin{align*}
\mathfrak{g}&=\mathfrak{a}\oplus\mathfrak{b}\cr
\mathfrak{g}^*&=\mathfrak{a}^0\oplus\mathfrak{b}^0
\end{align*}
of the Lie algebra and its dual, where $\left(\cdot\right)^0$ indicates annihilator. Let $\pi_{\mathfrak{a}^0},\pi_{\mathfrak{b}^0}$ be the projections associated with the decomposition on $\mathfrak{g}^*$; additionally, let us take
$\mu\in\mathfrak{b}^0,\nu\in\mathfrak{a}^0$. Then we can define the submanifold
\[
\tilde{\Lambda}_{\mu\nu}:=\left\{\left(g,\zeta\right)\in G\times\mathfrak{g}^*:\pi_{\mathfrak{b}^0}\left(\mbox{Ad}^*_{g^{-1}}\zeta\right)=\mu,\pi_{\mathfrak{a}^0}\left(\zeta\right)=\nu\right\}.
\]
We will see later that this is a first class submanifold, and that the kernel of the form induced by the canonical form on it by pullback is generated by the infinitesimal generators of the action of $A_\mu\times B_\nu$
on $G$ given by
\[
\left(a,b\right)\cdot g:=agb^{-1}
\]
lifted to the cotangent bundle. Here $A_\mu$ (resp. $B_\nu$) are the isotropy groups of $\mu$
(resp. $\nu$) for the actions
\begin{align*}
a\cdot\phi&:=\pi_{\mathfrak{b}^0}\left(\mbox{Ad}_{a^{-1}}^*\phi\right),\qquad\forall\phi\in\mathfrak{b}^0,\cr
b\cdot\psi&:=\pi_{\mathfrak{a}^0}\left(\mbox{Ad}_{b^{-1}}^*\psi\right),\qquad\forall\psi\in\mathfrak{a}^0.
\end{align*}
Quotient out by this action we can obtain a symplectic manifold, whose equations of motion are nothing but the equations of an AKS system.

\subsection{The non standard mechanics of AKS}
This section is inspired in the article of F\`{e}her \emph{et al.} \cite{feher-2002-301}. Let us consider as velocity space the bundle
\[
\Lambda:=TG\oplus\mathfrak{a}_G\oplus\mathfrak{b}_G
\]
where, as before, we suppose that $G$ admits a decomposition $G=AB$ and the fiber bundle on
$G$ denoted by $\mathfrak{a}_G$ and $\mathfrak{b}_G$ are simply $G\times\mathfrak{a}$
and $G\times\mathfrak{b}$; a point in this space (with the left trivialization for $TG$) will be
$\left(g,J;\alpha,\beta\right)$. There we will define the $1-$form $\mathfrak{g}-$valued
\[
\left.\theta\right|_{\left(g,J;\alpha,\beta\right)}:=\left.\lambda\right|_g-\left(J-\mbox{Ad}_{g^{-1}}\alpha-\beta\right)\dif t,
\]
(here $\lambda$ is the left Maurer-Cartan form) defining the prolongation structure, and we will consider the lagrangian
\[
\left.L_{\mu\nu}\right|_{\left(g,J;\alpha,\beta\right)}:=\left[\frac{1}{2}\mathsf{B}\left(J,J\right)-\mu\left(\alpha\right)-\nu\left(\beta\right)\right]\dif t
\]
where $\mu\in\mathfrak{b}^0,\nu\in\mathfrak{a}^0$ and, as pointed out before,
$\left(\cdot\right)^0$ is the annihilator of the corresponding subalgebra.
\newline
The Euler-Lagrange equations can be obtained performing variations of the action
\[
S_{\mu\nu}[\gamma]:=\int_{I}\gamma^*\left(L_{\mu\nu}\dif t\right)
\]
where
$\gamma:t\mapsto\left(t;g\left(t\right),J\left(t\right);\alpha\left(t\right),\beta\left(t\right)\right)$
represents a curve in $\Lambda$. Concretely, a variation of such a curve is a vector field
\[
\widehat{\delta\gamma}:\left(t;g,J;\alpha,\beta\right)\mapsto\left(0;g,\xi,J,\delta
  J;\alpha,\delta\alpha,\beta,\delta\beta\right)
\]
for certain funtions $\xi:\Lambda\rightarrow\mathfrak{g},\delta
J:\Lambda\rightarrow\mathfrak{g},\delta\alpha:\Lambda\rightarrow\mathfrak{a},\delta\beta:\Lambda\rightarrow\mathfrak{b}$
(we are considering here $TG$ left trivialized as before) such that
\[
\gamma^*\left(\mathcal{L}_{\widehat{\delta\gamma}}\theta\right)=0.
\]
From this last condition we will obtain that on $\gamma$ the maps describing the variation must obey the relation
\begin{equation}\label{VarsAdmisibles}
\dif\xi+\left[\lambda,\xi\right]-\left(\delta
  J-\mbox{Ad}_{g^{-1}}\left(\delta\alpha\right)-\left[\mbox{Ad}_{g^{-1}}\alpha,\xi\right]-\delta\beta\right)\dif t=0.
\end{equation}
Therefore
\[
\delta S_{\mu\nu}[\gamma]=\int_{I}\gamma^*\left(\widehat{\delta\gamma}\cdot L_{\mu\nu}\dif t\right)
\]
where
\begin{align*}
\widehat{\delta\gamma}\cdot L_{\mu\nu}&=\mathsf{B}\left(\delta J\dif
  t,J\right)-\left[\mu\left(\delta\alpha\right)+\nu\left(\delta\beta\right)\right]\dif
t\cr
&=\mathsf{B}\left(\dif\xi+\left[\lambda+\left(\mbox{Ad}_{g^{-1}}\alpha\right)\dif
    t,\xi\right],J\right)+\cr
&\qquad+\left[\mathsf{B}\left(\mbox{Ad}_{g^{-1}}\delta\alpha+\delta\beta,J\right)-\mu\left(\delta\alpha\right)-\nu\left(\delta\beta\right)\right]\dif t.
\end{align*}
Each $\phi\in\mathfrak{g}^*$ and
$\kappa\in\Omega^1\left(\Lambda,\mathfrak{g}\right)$ allow us to define the $1-$form
$\mathfrak{g}^*-$valued $\mbox{ad}_{\kappa}^{\sharp}\phi$ through
\[
\left(\mbox{ad}_{\kappa}^{\sharp}\phi\right)\left(\xi\right):=\phi\left(\left[\xi,\kappa\right]\right),\qquad\forall\xi\in\mathfrak{g};
\]
let us indicate with the symbol $\xi^\flat$ the form
$\mathsf{B}\left(\xi,\cdot\right)$, and let us denote by $\mbox{Ad}^\sharp:G\rightarrow\mathsf{Hom}\left(\mathfrak{g}^*\right)$ the coadjoint representation, defined through
\[
\left(\mbox{Ad}^\sharp_g\phi\right)\left(\xi\right)=\phi\left(\mbox{Ad}_{g^{-1}}\xi\right)
\]
for all $\xi\in\mathfrak{g},\phi\in\mathfrak{g}^*,g\in G$. With this notation, the element of $\mathfrak{g}^*$ obtained by contracting a vector $v$
tangent to $\Lambda$ with the $1-$form satisfy the equation
\[
\left(\mbox{ad}_{\kappa}^\sharp\phi\right)\left(v\right)=\mbox{ad}^\sharp_{\kappa\left(v\right)}\phi,
\]
where the element in the right hand side is the infinitesimal generator for the coadjoint action associated to $\kappa\left(v\right)\in\mathfrak{g}$ in the point $\phi\in\mathfrak{g}^*.$ Then the variations in $\delta\alpha,\delta\beta$ will give us the equations
\[
\pi_{\mathfrak{b}^0}\left(\mbox{Ad}_g^\sharp J^\flat-\mu\right)=0,\pi_{\mathfrak{a}^0}\left(J^\flat-\nu\right)=0
\]
respectively, and the variation along $TG$ (compatible with the restriction
\eqref{VarsAdmisibles}) yields to
\[
\dif J^\flat+\mbox{ad}^\sharp_{\left(\lambda+\mbox{Ad}_{g^{-1}}\alpha\dif t\right)}\left(J^\flat\right)=0.
\]
The solutions of the Euler-Lagrange equations are integral sections of
$I\times\Lambda$ for the EDS generated by
\[
\begin{cases}
\dif J^\flat-\mbox{ad}^\sharp_{\beta\dif t}\left(J^\flat\right),&\cr
\lambda-\left(J-\mbox{Ad}_{g^{-1}}\alpha-\beta\right)\dif t,&\cr
\pi_{\mathfrak{b}^0}\left(\mbox{Ad}_g^\sharp J^\flat\right)-\mu,&\cr
\pi_{\mathfrak{a}^0}\left(J^\flat\right)-\nu&
\end{cases}
\]
satisfying the independence condition $\dif t\not=0$; here we will assume that
$\mbox{ad}^\sharp_{\zeta\dif t}\zeta^\flat=0$ for all $\zeta\in\mathfrak{g}$.

\section{Hamiltonian problem}
Let us try to build a hamiltonian theory with these data. To this end, we will define the canonical Lepage equivalent to the variational problem
\[
\left(\Lambda,L_{\mu\nu}\dif t,\left<\theta\right>_{\text{diff}}\right).
\]
So if we set
$\tilde{\Lambda}:=I\times\left(\Lambda\oplus\left(G\times\mathfrak{g}^*\right)\right)$
and on it we define the $1-$form
\[
\left.\tilde{\lambda}\right|_{\left(g,J;\alpha,\beta;\sigma\right)}:=L_{\mu\nu}\left(g,J;\alpha,\beta\right)\dif t+\sigma\left(\theta\right),
\]
the EDS describing the extremals for the variational problem $\left(\tilde\Lambda,\tilde\lambda,0\right)$ will be generated by
\[
\begin{cases}
\dif\sigma+\mbox{ad}^\sharp_{\left(\lambda+\mbox{Ad}_{g^{-1}}\alpha\dif
    t\right)}\sigma&\cr
J^\flat-\sigma&\cr
\lambda-\left(J-\mbox{Ad}_{g^{-1}}\alpha-\beta\right)\dif t,&\cr
\pi_{\mathfrak{b}^0}\left(\mbox{Ad}_g^\sharp\sigma\right)-\mu,&\cr
\pi_{\mathfrak{a}^0}\left(\sigma\right)-\nu,&
\end{cases}
\]
where $\phi^\natural$ is the inverse isomorfism to $\left(\cdot\right)^\flat$. Therefore the projection
\[
\Pi:I\times\tilde{\Lambda}\rightarrow
I\times\Lambda:\left(t;g,J;\alpha,\beta;\sigma\right)\mapsto\left(t;g,J;\alpha,\beta\right)
\]
uniquely maps solutions of this EDS onto solutions of the Euler-Lagrange equations found above.
\newline
Then $L_{t_0}:=\left.\tilde{\Lambda}\right|\left\{t=t_0\right\}$ is a presymplectic manifold, with presymplectic form
\begin{align*}
\left.\omega_0\right|_{\left(g,J;\alpha,\beta;\sigma\right)}&:=\left.\dif\tilde{\lambda}\right|\left\{t=t_0\right\}\cr
&=\left<\dif\sigma\stackrel{\wedge}{,}\lambda\right>+\sigma\left(\dif\lambda\right)\cr
&=\left<\dif\sigma\stackrel{\wedge}{,}\lambda\right>-\frac{1}{2}\sigma\left(\left[\lambda\stackrel{\wedge}{,}\lambda\right]\right),
\end{align*}
where $\left<\cdot,\cdot\right>$ indicates the contraction of $\mathfrak{g}^*$
with $\mathfrak{g}$. The hamiltonian governing the dynamics is obtained from $\tilde{\lambda}$ via
\begin{align*}
H\left(g,J;\alpha,\beta;\sigma\right)&:=\left.\left(\partial_t\lrcorner\tilde{\lambda}\right)\right|\left\{t=t_0\right\}\cr
&=\sigma\left(\mbox{Ad}_{g^{-1}}\alpha+\beta-J\right)+L\left(g,J;\alpha,\beta\right)\cr
&=-\frac{1}{2}\sigma\left(J\right)+\frac{1}{2}\left(J^\flat-\sigma\right)\left(J\right)+\left(\mbox{Ad}^\sharp_g\sigma-\mu\right)\left(\alpha\right)+\left(\sigma-\nu\right)\left(\beta\right).
\end{align*}

\subsection{Gotay, Nester and Hinds algorithm}
Let us apply the Gotay, Nester and Hinds algorithm \cite{GotayNester} to the presymplectic manifold $\left(L_{t_0},\omega_0\right)$ with the Hamiltonian function $H$. For $l:=\left(g,J;\alpha,\beta;\sigma\right)$ we will have that
\begin{align*}
\left(T_lL_{t_0}\right)^\perp&=\left\{X=\left(\xi,\delta
    J,\delta\alpha,\delta\beta;\delta\sigma\right)\in
  T_lL_{t_0}:X\lrcorner\left.\omega_0\right|_l=0\right\}\cr
&=\left\{X\in
  T_lL_{t_0}:\delta\sigma\left(\lambda\right)-\dif\sigma\left(\xi\right)-\sigma\left(\left[\xi,\lambda\right]\right)=0\right\}\cr
&=\left\{X\in
  T_lL_{t_0}:\delta\sigma=0,\xi=0\right\},
\end{align*}
and so the invariant points will be $l\in L_{t_0}$ such that
$\left.X\right|_l\cdot H=0$ for all
$X\in\left(T_lL_{t_0}\right)^\perp$, that is
\begin{align*}
0&=\left.X\right|_l\cdot H\cr
&=\left(J^\flat-\sigma\right)\left(\delta J\right)+\left(\mbox{Ad}^\sharp_g\sigma-\mu\right)\left(\delta\alpha\right)+\left(\sigma-\nu\right)\left(\delta\beta\right)
\end{align*}
for all $\delta\alpha\in\mathfrak{a},\delta\beta\in\mathfrak{b},\delta
J\in\mathfrak{g}$. The primary constraints results
\[
\begin{cases}
J^\flat-\sigma=0&\cr
\pi_{\mathfrak{b}^0}\left(\mbox{Ad}_g^\sharp\sigma-\mu\right)=0&\cr
\pi_{\mathfrak{a}^0}\left(\sigma-\nu\right)=0.
\end{cases}
\]
Defining
\[
\tilde{\Lambda}^{\left(1\right)}:=\left\{\left(g,J;\alpha,\beta;\sigma\right)\in L_{t_0}:J^\flat=\sigma,\pi_{\mathfrak{b}^0}\left(\mbox{Ad}_g^\sharp\sigma-\mu\right)=0,\pi_{\mathfrak{a}^0}\left(\sigma-\nu\right)=0\right\},
\]
we want to see if this submanifold is invariant respect to the flux generated by our hamiltonian. In order to accomplish it, we will use the following result.

\begin{lem}\label{LemaBasePerp}
$\left(T_l\tilde{\Lambda}^{\left(1\right)}\right)^\perp\subset
T_lL_{t_0}$ is generated by the hamiltonian vector fields associated to the functions
\begin{align*}
f_2^\kappa\left(g,J;\alpha,\beta;\sigma\right)&:=\left(\pi_{\mathfrak{b}^0}\left(\mbox{Ad}_g^\sharp\sigma\right)-\mu\right)\left(\kappa\right)\cr
f_2^\omega\left(g,J;\alpha,\beta;\sigma\right)&:=\left(\pi_{\mathfrak{a}^0}\left(\sigma\right)-\nu\right)\left(\omega\right)
\end{align*}
where $\kappa,\omega\in\mathfrak{g}$ are arbitrary elements.
\end{lem}
\begin{proof}
The idea of the proof is to calculate the dimension of
$\left(T_l\tilde{\Lambda}^{\left(1\right)}\right)^\perp$ using the formula
\[
\mbox{dim}\,\left(TN\right)^\perp=\mbox{dim}\,P-\mbox{dim}\,N+\mbox{dim}\,\left(\mbox{Ker}\,\omega_0\cap
TN\right)
\]
valid for every submanifold $N\subset P$ of a presymplectic manifold. Let us note that
$\mbox{Ker}\,\left.\omega_0\right|_l=\left(T_lL_{t_0}\right)^\perp$, so if
$X=\left(\xi,\delta J,\delta\alpha,\delta\beta,\delta\sigma\right)$ is in
$\left(T_lL_{t_0}\right)^\perp$ and satisfies the constraint defined by the
$f_1$'s, then $X=\left(0,0,\delta\alpha,\delta\beta,0\right)$, and so it is tangent to $\tilde{\Lambda}^{\left(1\right)}$ at $l$. So
\[
\mbox{dim}\,\left(\mbox{Ker}\,\left.\omega_0\right|_l\cap T_l\tilde{\Lambda}^{\left(1\right)}\right)=\mbox{dim}\,\mathfrak{a}+\mbox{dim}\,\mathfrak{b}=\mbox{dim}\,\mathfrak{g}.
\]
It remains to see if there exists some relationship between the maps $f_2$ and
$f_3$; its differentials are
\begin{align*}
\dif
f_2^\kappa&=\left(\pi_{\mathfrak{b}^\perp}\left(\mbox{Ad}^\sharp_g\left(\dif\sigma+\mbox{ad}^\sharp_\lambda\sigma\right)\right)\right)\left(\kappa\right)\cr
&=\dif\sigma\left(\mbox{Ad}_{g^{-1}}\pi_{\mathfrak{a}}\left(\kappa\right)\right)-\left(\mbox{ad}^\sharp_{\mbox{Ad}_{g^{-1}}\pi_{\mathfrak{a}}\left(\kappa\right)}\sigma\right)\left(\lambda\right)\cr
\dif f_3^\omega&=\dif\sigma\left(\pi_{\mathfrak{b}}\left(\omega\right)\right).
\end{align*}
We want to note that if, for example, $\sigma$ is a central element, then it could happened that these differentials will be linearly dependent. We can therefore assume the following hypothesis, valid for example when te factorization comes from the Iwasawa decomposition:
\begin{quote}
\emph{For any $g\in G$, if
  $X\in\mathfrak{b}\cap\mbox{Ad}_g\mathfrak{a}$, then $X=0$.}
\end{quote}
We are then sure that our functions are linearly independent; then there are $\mbox{dim}\,\mathfrak{a}$ independent functions of type $f_2$ and
$\mbox{dim}\,\mathfrak{b}$ of $f_3$ type. Therefore
\[
\mbox{dim}\,\tilde{\Lambda}^{\left(1\right)}=\mbox{dim}\,L_{t_0}-\mbox{dim}\,\mathfrak{g}-\mbox{dim}\,\mathfrak{a}-\mbox{dim}\,\mathfrak{b},
\]
and putting all these thing together
\[
\mbox{dim}\,\left(TN\right)^\perp=2\mbox{dim}\,\mathfrak{g}+\mbox{dim}\,\mathfrak{a}+\mbox{dim}\,\mathfrak{b}=3\mbox{dim}\,\mathfrak{g}.
\]
Finally by writing $X_{f^\phi_k}:=\left(\xi^\phi_k,\delta
  J^\phi_k,\delta\alpha^\phi_k,\delta\beta^\phi_k,\delta\sigma^\phi_k\right)$
for the hamiltonian vector field associated to the function $f^\phi_k$,
it results that
\[
X_{f^\phi_k}\lrcorner\omega_0=\left(\delta\sigma^\phi_k+\mbox{ad}_{\xi^\phi_k}^\sharp\sigma\right)\left(\lambda\right)-\dif\sigma\left(\xi^\phi_k\right).
\]
This will implies the following expressions for the hamiltonian fields:
\begin{align*}
X_{f_2^\kappa}&=\left(-\mbox{Ad}_{g^{-1}}\pi_{\mathfrak{a}}\left(\kappa\right),\delta
  J^\kappa_2,\delta\gamma^\kappa_2,\delta\chi^\kappa_2,0\right)\cr
X_{f_3^\omega}&=\left(\pi_{\mathfrak{b}}\left(\omega\right),\delta J^\omega_3,\delta\gamma^\omega_3,\delta\chi^\omega_3,-\mbox{ad}^\sharp_{\pi_{\mathfrak{b}}\left(\omega\right)}\sigma\right)
\end{align*}
(where $\delta\gamma_k^\phi,\delta\chi_k^\phi$ are arbitrary $\mathfrak{a}$ and
$\mathfrak{b}-$valued functions, $\delta J_k^\phi$
a $\mathfrak{g}-$valued function also arbitrary) while the maps $f_1$ has no associated hamiltonian vector fields (we are working with a presymplectic structure!). As we have
$3\mbox{dim}\,\mathfrak{g}$ linearly independent vector fields here, and they belong to $\left(T\tilde{\Lambda}^{\left(1\right)}\right)^\perp$,
we have proved the lemma.
\end{proof}

Let us now apply the invariance condition; this involves to find the points $l\in\tilde{\Lambda}^{\left(1\right)}$ where the components of $\left.\dif H\right|_l$ in the directions of $\left(T_lL_{t_0}\right)^\perp$ annihilates. Now if $l=\left(g,J;\alpha,\beta;\sigma\right)\in\tilde{\Lambda}^{\left(1\right)}$
\begin{align*}
\left.X_{f_2^\kappa}\right|_l\cdot H&=-\frac{1}{2}\sigma\left(\delta J_2^\kappa\right)+\frac{1}{2}\mathsf{B}\left(\delta J_2^\kappa,J\right)-\mbox{Ad}^\sharp_g\left(\mbox{ad}^\sharp_{\mbox{Ad}_{g^{-1}}\pi_{\mathfrak{a}}\left(\kappa\right)}J\right)\left(\alpha\right)\cr
&=\sigma\left(\left[\mbox{Ad}_{g^{-1}}\pi_{\mathfrak{a}}\left(\kappa\right),\mbox{Ad}_{g^{-1}}\alpha\right]\right)\cr
&=\left(\mbox{Ad}_{g}^\sharp\sigma\right)\left(\left[\pi_{\mathfrak{a}}\left(\kappa\right),\alpha\right]\right)\cr
&=\left(\pi_{\mathfrak{b}^0}\left(\mbox{Ad}_{g}^\sharp\sigma\right)\right)\left(\left[\pi_{\mathfrak{a}}\left(\kappa\right),\alpha\right]\right)\cr
&=\mu\left(\left[\pi_{\mathfrak{a}}\left(\kappa\right),\alpha\right]\right)\cr
&=\left(\mbox{ad}^\sharp_{\alpha}\mu\right)\left(\pi_{\mathfrak{a}}\left(\kappa\right)\right),
\end{align*}
so the stability of the constraints defined by $f_2$ forces $\alpha$ to belong to $\mathfrak{a}_\mu$, the isotropy group of the element $\mu\in\mathfrak{b}^0$ respect to the action\footnote{It is a true action!}
\[
a\in A\mapsto\pi_{\mathfrak{b}^0}\left(\mbox{Ad}^\sharp_a\mu\right).
\]
Moreover for the constraints $f_3$, and supposing that $\mbox{ad}^\sharp_\zeta\zeta^\flat=0$ for all $\zeta\in\mathfrak{g}$
\begin{align*}
\left.X_{f_3^\omega}\right|_l\cdot H&=-\frac{1}{2}\sigma\left(\delta J_3^\omega\right)+\frac{1}{2}\mathsf{B}\left(\delta J_3^\omega,J\right)+\frac{1}{2}\left(\mbox{ad}^\sharp_{\pi_{\mathfrak{b}}\left(\omega\right)}\sigma\right)\left(J\right)+\cr
&\qquad+\frac{1}{2}\left(\mbox{ad}^\sharp_{\pi_{\mathfrak{b}}\left(\omega\right)}\sigma\right)\left(J\right)-\left(\mbox{Ad}^\sharp_g\mbox{ad}^\sharp_{\pi_{\mathfrak{b}}\left(\omega\right)}\sigma\right)\left(\alpha\right)-\cr
&\qquad\qquad-\left(\mbox{ad}^\sharp_{\pi_{\mathfrak{b}}\left(\omega\right)}\sigma\right)\left(\beta\right)+\left(\mbox{Ad}^\sharp_g\mbox{ad}^\sharp_{\pi_{\mathfrak{b}}\left(\omega\right)}\sigma\right)\left(\alpha\right)\cr
&=-\left(\mbox{ad}^\sharp_{\pi_{\mathfrak{b}}\left(\omega\right)}\sigma\right)\left(\beta\right)\cr
&=\sigma\left(\left[\pi_{\mathfrak{b}}\left(\omega\right),\beta\right]\right)\cr
&=\pi_{\mathfrak{a}^0}\left(\sigma\right)\left(\left[\pi_{\mathfrak{b}}\left(\omega\right),\beta\right]\right)\cr
&=\left(\mbox{ad}^\sharp_{\beta}\left(\pi_{\mathfrak{a}^0}\left(J\right)\right)\right)\left(\pi_{\mathfrak{b}}\left(\omega\right)\right)\cr
&=\left(\pi_{\mathfrak{a}^0}\left(\mbox{ad}^\sharp_{\beta}\nu\right)\right)\left(\pi_{\mathfrak{b}}\left(\omega\right)\right)
\end{align*}
deducing that $\beta$ must live in the Lie algebra of the isotropy group of $\nu$ respect to the action\footnote{See previous footnote!}
\[
b\in B\mapsto\pi_{\mathfrak{a}^0}\left(\mbox{Ad}^\sharp_b\nu\right)
\]
for all $\nu\in\mathfrak{a}^0$.
\newline
Then it will arise new constraints, the secondary constraints surface being
\[
\tilde{\Lambda}^{\left(2\right)}:=\{\left(g,J;\alpha,\beta;\sigma\right)\in\tilde{\Lambda}^{\left(1\right)}:\alpha\in\mathfrak{a}_{\mu},\beta\in\mathfrak{b}_{\nu}\}.
\]
So we have the following result.
\begin{prop}
$\tilde{\Lambda}^{\left(2\right)}$ is invariant.
\end{prop}
\begin{proof}
Let us take $l:=\left(g,\sigma;\alpha,\beta,\sigma\right)\in\tilde{\Lambda}^{\left(2\right)}$. Therefore $\left(T_lL_{t_0}\right)^\perp=\{\left(\xi,\delta J,\delta\alpha,\delta\beta;\delta\sigma\right)\in T_lL_{t_0}:\delta\sigma=\xi=0\}$; then an element $X$ of $\left(T_lL_{t_0}\right)^\perp\cap\left(T_l\tilde{\Lambda}^{\left(2\right)}\right)$ can be written as
\[
X=\left(0,0;\delta\alpha,\delta\beta;0\right)
\]
with $\delta\alpha\in\mathfrak{a}_\mu,\delta\beta\in\mathfrak{b}_\nu$. This means that $\mbox{dim}\,\left(\left(T_lL_{t_0}\right)^\perp\cap\left(T_l\tilde{\Lambda}^{\left(2\right)}\right)\right)=\mbox{dim}\,\mathfrak{a}_\mu+\mbox{dim}\,\mathfrak{b}_\nu$. Moreover
\[
\mbox{dim}\,\left(\tilde{\Lambda}^{\left(2\right)}\right)=\mbox{dim}\,\left(\tilde{\Lambda}^{\left(1\right)}\right)-\left(\mbox{dim}\,\left(\mathfrak{a}\right)-\mbox{dim}\,\left(\mathfrak{a}_\mu\right)+\mbox{dim}\,\left(\mathfrak{b}\right)-\mbox{dim}\,\left(\mathfrak{b}_\nu\right)\right)
\]
and thus
\[
\mbox{dim}\,\left(T_l\tilde{\Lambda}^{\left(2\right)}\right)^\perp=\mbox{dim}\,\left(L_{t_0}\right)-\mbox{dim}\,\left(\tilde{\Lambda}^{\left(1\right)}\right)+\mbox{dim}\,\left(\mathfrak{a}\right)+\mbox{dim}\,\left(\mathfrak{b}\right).
\]
As we see before $\mbox{dim}\,\left(L_{t_0}\right)-\mbox{dim}\,\left(\tilde{\Lambda}^{\left(1\right)}\right)=2\mbox{dim}\,\mathfrak{g}$, and so $\mbox{dim}\,\left(T_l\tilde{\Lambda}^{\left(2\right)}\right)^\perp=3\mbox{dim}\,\mathfrak{g}$. We know that $\left(T_l\tilde{\Lambda}^{\left(1\right)}\right)^\perp\subset\left(T_l\tilde{\Lambda}^{\left(2\right)}\right)^\perp$ because $\tilde{\Lambda}^{\left(2\right)}$ is a submanifold of $\tilde{\Lambda}^{\left(1\right)}$; both have the same dimension, so they will be equal. In the lemma \ref{LemaBasePerp} we obtained the basis
\begin{multline*}
\bigg\{X_1^\kappa:=\left(-\mbox{Ad}_{g^{-1}}\kappa,0;0,0;0\right),X_2^\omega:=\left(\omega,0;0,0;-\mbox{ad}^\sharp_\omega\sigma\right),\widehat{\delta J}:=\left(0,\delta J;0,0;0\right),\cr
,\widehat{\delta\alpha}:=\left(0,0;\delta\alpha,0;0\right),\widehat{\delta\beta}:=\left(0,0;0,\delta\beta;0\right):\kappa,\delta\alpha\in\mathfrak{a},\omega,\delta\beta\in\mathfrak{b},\delta J\in\mathfrak{g}\bigg\}
\end{multline*}
for $\left(T_l\tilde{\Lambda}^{\left(1\right)}\right)^\perp$. Then there are no more constraints, because $\tilde{\Lambda}^{\left(2\right)}$ is defined as the subset of $\tilde{\Lambda}^{\left(1\right)}$ where $H$ is invariant by elements of $\left(T\tilde{\Lambda}^{\left(1\right)}\right)^{\perp}$.
\end{proof}

\subsection{Gauge fixing and reduction}
From the calculations made above we have that $\tilde{\Lambda}^{\left(2\right)}$ is a presymplectic manifold, with presymplectic form $\omega_2:=\left.\omega_0\right|\tilde{\Lambda}^{\left(2\right)}$ and this form has the kernel
\begin{align*}
\mbox{Ker}\,\left.\omega_2\right|_l&=\left(T_l\tilde{\Lambda}^{\left(2\right)}\right)^\perp\cap\left(T_l\tilde{\Lambda}^{\left(2\right)}\right).
\end{align*}
Now an arbitrary element of $\left(T_l\tilde{\Lambda}^{\left(2\right)}\right)^\perp$ can be written as
\[
X^\perp:=\left(-\mbox{Ad}_{g^{-1}}\kappa+\omega,\delta J;\delta\alpha,\delta\beta;-\mbox{ad}_\omega^\sharp\sigma\right)
\]
for some $\kappa,\delta\alpha\in\mathfrak{a},\omega,\delta\beta\in\mathfrak{b}$ y $\delta J\in\mathfrak{g}$. Because
\begin{multline*}
T_l\tilde{\Lambda}^{\left(2\right)}=\bigg\{\left(\xi,\delta\sigma;\delta\alpha,\delta\beta;\delta\sigma\right):\delta\alpha\in\mathfrak{a}_\mu,\delta\beta\in\mathfrak{b}_\nu,\cr
,\pi_{\mathfrak{b}^0}\left(\mbox{Ad}^\sharp_g\left(\delta\sigma+\mbox{ad}_\xi^\sharp\sigma\right)\right)=0,\pi_{\mathfrak{a}^0}\left(\delta\sigma\right)=0\bigg\}
\end{multline*}
then $X^\perp$ will belong to the intersection of these spaces if and only if $\delta\alpha\in\mathfrak{a}_\mu,\delta\beta\in\mathfrak{b}_\nu$, $\left(\delta J\right)^\flat=-\mbox{ad}^\sharp_\omega\sigma$ and moreover
\begin{align*}
0&=\pi_{\mathfrak{b}^0}\left(\mbox{Ad}^\sharp_g\left(-\mbox{ad}^\sharp_\omega\sigma+\mbox{ad}_{-\mbox{Ad}_{g^{-1}}\kappa+\omega}^\sharp\sigma\right)\right)\cr
&=-\pi_{\mathfrak{b}^0}\left(\mbox{Ad}^\sharp_g\left(\mbox{ad}_{\mbox{Ad}_{g^{-1}}\kappa}^\sharp\sigma\right)\right)\cr
&=-\pi_{\mathfrak{b}^0}\left(\mbox{ad}_{\kappa}^\sharp\left(\mbox{Ad}_{g}^\sharp\sigma\right)\right)\cr
0&=-\pi_{\mathfrak{a}^0}\left(\mbox{ad}^\sharp_\omega\sigma\right).
\end{align*}
Taking into account that $\mathfrak{a}^0$ is invariant by the $\mbox{ad}^\sharp-$action of $\mathfrak{a}$, and the same is true for $\mathfrak{b}^0$ respect to $\mathfrak{b}$, these conditions are equivalent to
\begin{align*}
0&=-\pi_{\mathfrak{b}^0}\left(\mbox{ad}_{\kappa}^\sharp\pi_{\mathfrak{b}^0}\left(\mbox{Ad}_{g}^\sharp\sigma\right)\right)\cr
&=-\pi_{\mathfrak{b}^0}\left(\mbox{ad}_{\kappa}^\sharp\mu\right)\cr
0&=-\pi_{\mathfrak{a}^0}\left(\mbox{ad}^\sharp_\omega\pi_{\mathfrak{a}^0}\sigma\right)\cr
&=-\pi_{\mathfrak{a}^0}\left(\mbox{ad}^\sharp_\omega\nu\right)
\end{align*}
because $\kappa\in\mathfrak{a},\omega\in\mathfrak{b}$. Therefore $\kappa\in\mathfrak{a}_\mu,\omega\in\mathfrak{b}_\nu$ ensures us, together with the conditions found above, that $X^\perp$ is tangent to $\tilde{\Lambda}^{\left(2\right)}$.
We can interpret the singular directions  $\widehat{\delta\alpha},\widehat{\delta\beta}$ as infinitesimal generators of an action of the abelian group $\mathfrak{a}_\mu\times\mathfrak{b}_\nu$ on $\tilde{\Lambda}^{\left(2\right)}$ via
\[
\left(\alpha',\beta'\right)\cdot\left(g,\sigma;\alpha,\beta;\sigma\right)=\left(g,\sigma;\alpha+\alpha',\beta+\beta';\sigma\right).
\]
By defining $\Lambda_{\mu\nu}:=\left\{\left(g,\sigma\right)\in G\times\mathfrak{g}:\pi_{\mathfrak{b}^0}\left(\mbox{Ad}_g^\sharp\sigma\right)=\mu,\pi_{\mathfrak{a}^0}\left(\sigma\right)=\nu\right\}$ we see that the projection $\Pi':\left(g,\sigma;\alpha,\beta;\sigma\right)\in\tilde{\Lambda}^{\left(2\right)}\mapsto\left(g,\sigma\right)\in\Lambda_{\mu\nu}$ induces a diffeomorphism between the presymplectic manifolds $\tilde{\Lambda}^{\left(2\right)}/\left(\mathfrak{a}_\mu\times\mathfrak{b}_\nu\right)$ and $\Lambda_{\mu\nu}$, the latter with the structure induced by considering it as submanifold of $G\times\mathfrak{g}^*$ with symplectic structure
\[
\omega:=\left<\dif\sigma\stackrel{\wedge}{,}\lambda\right>+\sigma\left(\dif\lambda\right),
\]
which is the canonical one in $G\times\mathfrak{g}^*$. Then $\Lambda_{\mu\nu}$ is a presymplectic map with the singular directions in $\left(g,\sigma\right)$ generated by vectors 
\[
X\left(\kappa,\omega\right):=\left(\mbox{Ad}_{g^{-1}}\kappa-\omega,\mbox{ad}^\sharp_\omega\sigma\right)
\]
for all
$\left(\kappa,\omega\right)\in\mathfrak{a}_\mu\times\mathfrak{b}_\nu$.

\begin{lem}
Let $\xi^{L,R}$ be the infinitesimal genrator associated to the lift of the left traslations (for the superindex $L$) and the right traslations (for the other superindex $R$) in $G$ to $G\times\mathfrak{g}^*$, considered as the cotangent bundle of $G$ via left trivialization. Then $X\left(\kappa,\omega\right)=\left.\left(\kappa^L+\omega^R\right)\right|_{\left(g,\sigma\right)}$. 
\end{lem}

This lemma gives sense to the AKS scheme from Marsden-Weinstein reduction viewpoint: The group $G$ can be considered as a $A\times B$-space via the action
\[
A\times B\times G\rightarrow G:\left(a,b;g\right)\mapsto agb^{-1},
\]
and so $G\times\mathfrak{g}^*$, if it is identified with $T^*G$ via left trivialization. It is immediate to calculate the momentum map using the fact that our symplectic structure is exact and the action is lifted; we obtain that
\begin{align*}
J:&G\times\mathfrak{g}^*\rightarrow\mathfrak{b}^0\times\mathfrak{a}^0\cr
&\left(g,\sigma\right)\mapsto\left(\pi_{\mathfrak{b}^0}\left(\mbox{Ad}^\sharp_g\sigma\right),\pi_{\mathfrak{a}^0}\left(\sigma\right)\right).
\end{align*}
Thus we will have that $\Lambda_{\mu\nu}=J^{-1}\left(\mu,\nu\right)$; using reduction theory (see \cite{A-M}) the map from $\Lambda_{\mu\nu}$ to $\Lambda_{\mu\nu}/A_\mu\times B_\nu$ is presymplectic, and the solution curves of the dynamical system defined there by the invaint hamiltonian $H\left(g,\sigma\right)=\frac{1}{2}\sigma\left(\sigma^\flat\right)$ are in $1-1$ relation with the evolution curves of the system induced on the quotient. To work with these equations, let us introduce a convenient coordinate system; now the map
\begin{align*}
M_{\mu\nu}:\Lambda_{\mu\nu}\rightarrow\mathcal{O}^A_\mu\times\mathcal{O}^B_\nu:\left(g,\sigma\right)&\mapsto\left(\pi_{\mathfrak{b}^0}\left(\mbox{Ad}^\sharp_{g_A^{-1}}\mu\right),\pi_{\mathfrak{a}^0}\left(\mbox{Ad}^\sharp_{g_B}\nu\right)\right)\cr
&\mapsto\left(\pi_{\mathfrak{b}^0}\left(\mbox{Ad}^\sharp_{g_B}\sigma\right),\pi_{\mathfrak{a}^0}\left(\mbox{Ad}^\sharp_{g_B}\sigma\right)\right)
\end{align*}
iff $g=g_Ag_B$, induces a diffeomorphism in the quotient $\Lambda_{\mu\nu}/\left(A_\mu\times B_\nu\right)$. So if $\left(g,\xi;\sigma,\eta\right)$ is a vector tangent to $G\times\mathfrak{g}^*$ (as always, the trivializations are left trivializations) the derivative of $M_{\mu\nu}$ has the following form
\begin{multline}\label{DerivadaLmunu}
\left.\left(M_{\mu\nu}\right)_*\right|_{\left(g,\sigma\right)}\left(g,\sigma;\xi,\eta\right)=\left(-\pi_{\mathfrak{b}^0}\left(\mbox{ad}^\sharp_{\xi_A}\mbox{Ad}^\sharp_{g_A^{-1}}\mu\right),\pi_{\mathfrak{a}^0}\left(\mbox{ad}^\sharp_{\xi_B}\mbox{Ad}^\sharp_{g_B}\nu\right)\right)\cr
=\left(\pi_{\mathfrak{b}^0}\left(\mbox{ad}^\sharp_{\xi_B}\mbox{Ad}^\sharp_{g_B}\sigma\right)+\pi_{\mathfrak{b}^0}\left(\mbox{Ad}^\sharp_{g_B}\eta\right),\pi_{\mathfrak{a}^0}\left(\mbox{ad}^\sharp_{\xi_B}\mbox{Ad}^\sharp_{g_B}\sigma\right)+\pi_{\mathfrak{a}^0}\left(\mbox{Ad}^\sharp_{g_B}\eta\right)\right)
\end{multline}
if and only if $g=g_Ag_B,\xi_A=\pi_{\mathfrak{a}}\left(\mbox{Ad}_{g_B}\xi\right),\xi_B=\pi_{\mathfrak{b}}\left(\mbox{Ad}_{g_B}\xi\right)$. So we have the next surprinsing result.

\begin{prop}
Let $\mathcal{O}^A_\mu\times\mathcal{O}^B_\nu$ be the symplectic manifold dressed with the product symplectic structure $\omega_{\mu\nu}=\omega_\mu-\omega_\nu$, where $\omega_{\mu,\nu}$ are the corresponding Kirillov-Kostant symplectic structures on every orbit. If $i_{\mu\nu}:\Lambda\hookrightarrow G\times\mathfrak{g}^*$ is the immersion, then 
\[
i_{\mu\nu}^*\omega=M_{\mu\nu}^*\omega_{\mu\nu}.
\]
\end{prop}
\begin{proof}
If $\left(\omega_1,\omega_2\right)=\left(\pi_{\mathfrak{b}^0}\left(\mbox{Ad}^\sharp_a\mu\right),\pi_{\mathfrak{a}^0}\left(\mbox{Ad}^\sharp_b\nu\right)\right)$ is an arbitrary element of $\mathcal{O}^A_\mu\times\mathcal{O}^B_\nu$, then the tangent space in this point is given by
\[
T_{\left(\omega_1,\omega_2\right)}\left(\mathcal{O}^A_\mu\times\mathcal{O}^B_\nu\right)=\left\{\left(\pi_{\mathfrak{b}^0}\left(\mbox{ad}^\sharp_\xi\omega_1\right),\pi_{\mathfrak{a}^0}\left(\mbox{ad}^\sharp_\zeta\omega_2\right)\right):\xi\in\mathfrak{a},\zeta\in\mathfrak{b}\right\}.
\]
The symplectic structure $\omega_{\mu\nu}$ in these terms can be written as
\begin{multline}\label{EstrSimplecticaSobreProductoOrbitas}
\left.\omega_{\mu\nu}\right|_{\left(\omega_1,\omega_2\right)}\left(\left(\pi_{\mathfrak{b}^0}\left(\mbox{ad}^\sharp_{\xi_1}\omega_2\right),\pi_{\mathfrak{a}^0}\left(\mbox{ad}^\sharp_{\zeta_1}\omega_2\right)\right),\left(\pi_{\mathfrak{b}^0}\left(\mbox{ad}^\sharp_{\xi_2}\omega_2\right),\pi_{\mathfrak{a}^0}\left(\mbox{ad}^\sharp_{\zeta_2}\omega_2\right)\right)\right)=\cr
=\omega_1\left(\left[\xi_1,\xi_2\right]\right)-\omega_2\left(\left[\zeta_1,\zeta_2\right]\right).
\end{multline}
Let $\left(g,\sigma;\xi,\eta\right)$ be a vector tangent to $\Lambda_{\mu\nu}$ at $\left(g,\sigma\right)$; then we have that
\begin{equation}\label{CondTangencyLambdamunu}
\begin{cases}
\pi_{\mathfrak{b}^0}\left(\mbox{Ad}^\sharp_g\left(\mbox{ad}^\sharp_\xi\sigma+\eta\right)\right)=0&\cr
\pi_{\mathfrak{a}^0}\left(\eta\right)=0.
\end{cases}
\end{equation}
Because of the factorization $g=g_Ag_B$ the first condition can be written as
\[
\pi_{\mathfrak{b}^0}\left(\mbox{Ad}^\sharp_{g_A}\left(\mbox{ad}^\sharp_{\left(\mbox{Ad}_{g_B}\xi\right)}\left(\mbox{Ad}^\sharp_{g_B}\sigma\right)+\mbox{Ad}^\sharp_{g_B}\eta\right)\right)=0
\]
and so, from the nodegeneracy condition\footnote{That is, $\mbox{Ad}_g^\sharp\mathfrak{a}^0\cap\mathfrak{b}^0=0$ for all $g\in G$.}
\begin{equation}\label{CondTangencyLambdamunu2}
\pi_{\mathfrak{b}^0}\left(\mbox{ad}^\sharp_{\left(\mbox{Ad}_{g_B}\xi\right)}\left(\mbox{Ad}^\sharp_{g_B}\sigma\right)+\mbox{Ad}^\sharp_{g_B}\eta\right)=0.
\end{equation}
Letus now suppose that $\left(g,\sigma;\xi_1,\eta_1\right),\left(g,\sigma;\xi_2,\eta_2\right)\in T_{\left(g,\sigma\right)}\Lambda_{\mu\nu}$; then contracting them with th canonical form we will see that
\begin{multline*}
\left.\omega\right|_{\left(g,\sigma\right)}\left(\left(g,\sigma;\xi_1,\eta_1\right),\left(g,\sigma;\xi_2,\eta_2\right)\right)=\eta_1\left(\xi_2\right)-\eta_2\left(\xi_1\right)-\sigma\left(\left[\xi_1,\xi_2\right]\right)\cr
=\pi_{\mathfrak{b}^0}\left(\mbox{Ad}^\sharp_{g_B}\eta_1\right)\left(\left(\xi_2\right)_A\right)+\pi_{\mathfrak{a}^0}\left(\mbox{Ad}^\sharp_{g_B}\eta_1\right)\left(\left(\xi_2\right)_B\right)-\cr
-\pi_{\mathfrak{b}^0}\left(\mbox{Ad}^\sharp_{g_B}\eta_2\right)\left(\left(\xi_1\right)_A\right)-\pi_{\mathfrak{a}^0}\left(\mbox{Ad}^\sharp_{g_B}\eta_2\right)\left(\left(\xi_1\right)_A\right)-\cr
-\left(\mbox{Ad}^\sharp_{g_B}\sigma\right)\left(\left[\xi_1,\left(\xi_2\right)_A\right]\right)-\left(\mbox{Ad}^\sharp_{g_B}\sigma\right)\left(\left[\left(\xi_1\right)_B,\left(\xi_2\right)_B\right]\right)-\left(\mbox{Ad}^\sharp_{g_B}\sigma\right)\left(\left[\left(\xi_1\right)_A,\left(\xi_2\right)_B\right]\right)\cr
=\pi_{\mathfrak{b}^0}\left(\mbox{Ad}^\sharp_{g_B}\eta_1+\mbox{ad}^\sharp_{\xi_1}\left(\mbox{Ad}^\sharp_{g_B}\sigma\right)\right)\left(\left(\xi_2\right)_A\right)+\pi_{\mathfrak{a}^0}\left(\mbox{Ad}^\sharp_{g_B}\eta_1\right)\left(\left(\xi_2\right)_B\right)-\cr
-\pi_{\mathfrak{b}^0}\left(\mbox{Ad}^\sharp_{g_B}\eta_2+\mbox{ad}^\sharp_{\left(\xi_2\right)_B}\left(\mbox{Ad}^\sharp_{g_B}\sigma\right)\right)\left(\left(\xi_1\right)_A\right)-\pi_{\mathfrak{a}^0}\left(\mbox{Ad}^\sharp_{g_B}\eta_2\right)\left(\left(\xi_1\right)_A\right)-\cr
-\left(\mbox{Ad}^\sharp_{g_B}\sigma\right)\left(\left[\left(\xi_1\right)_B,\left(\xi_2\right)_B\right]\right)
\end{multline*}
where it was used that $\left(\xi_k\right)_{A,B}=\pi_{\mathfrak{a},\mathfrak{b}}\left(\mbox{Ad}_{g_B}\xi_k\right)$ for $k=1,2$. The first term in this expression annihilates because of Eq. \eqref{CondTangencyLambdamunu2}; the second an fourth also annihilates because from the second equation in Eq. \eqref{CondTangencyLambdamunu}, it follows that $\eta_k\in\mathfrak{b}^0,k=1,2$, and this subspace is invariant for the $B$-trasadjoint action $\mbox{Ad}^\sharp$. Moreover, by using Eq. \eqref{CondTangencyLambdamunu} again, the third term in the right hand side can be written as
\[
\pi_{\mathfrak{b}^0}\left(\mbox{Ad}^\sharp_{g_B}\eta_2+\mbox{ad}^\sharp_{\left(\xi_2\right)_B}\left(\mbox{Ad}^\sharp_{g_B}\sigma\right)\right)\left(\left(\xi_1\right)_A\right)=-\left(\mbox{Ad}^\sharp_{g_B}\sigma\right)\left(\left[\left(\xi_1\right)_A,\left(\xi_2\right)_A\right]\right),
\]
and so
\begin{multline*}
\left.\omega\right|_{\left(g,\sigma\right)}\left(\left(g,\sigma;\xi_1,\eta_1\right),\left(g,\sigma;\xi_2,\eta_2\right)\right)=\cr
=\pi_{\mathfrak{b}^0}\left(\mbox{Ad}^\sharp_{g_B}\sigma\right)\left(\left[\left(\xi_1\right)_A,\left(\xi_2\right)_A\right]\right)-\pi_{\mathfrak{a}^0}\left(\mbox{Ad}^\sharp_{g_B}\sigma\right)\left(\left[\left(\xi_1\right)_B,\left(\xi_2\right)_B\right]\right).
\end{multline*}
rom Eq. \eqref{EstrSimplecticaSobreProductoOrbitas} and using formula \eqref{DerivadaLmunu} for the derivative along the map $M_{\mu\nu}$, we have proved the proposition.
\end{proof}
Therefore the manifold $\mathcal{O}^A_\mu\times\mathcal{O}^B_\nu$ with the symplectic product structure $\omega_{\mu\nu}$ is symplectomorphic to the Marsden-Weinstein reduced space associated to the $A\times B$-action defined before on $G\times\mathfrak{g}^*$. As was said before, the hamiltonian $H\left(g,\sigma\right):=\frac{1}{2}\sigma\left(\sigma^\flat\right)$ is invariant by this action, so the solution curves of the dynamical system defined by such hamiltonian on $G\times\mathfrak{g}^*$ are in one-to-one correspondence with the solution curves on the system induced on $\mathcal{O}^A_\mu\times\mathcal{O}^B_\nu$ by the hamiltonian $H_{\mu\nu}$ \cite{A-M} defined through
\[
i_{\mu\nu}^*H=M_{\mu\nu}^*H_{\mu\nu}.
\]
Let us note now that if $M_{\mu\nu}\left(g,\sigma\right)=\left(\omega_1,\omega_2\right)$, then $\mbox{Ad}^\sharp_{g_B}\sigma=\omega_1+\omega_2$, and from here
\begin{align*}
H_{\mu\nu}\left(\omega_1,\omega_2\right)&=\frac{1}{2}\left(\mbox{Ad}_{g_B^{-1}}^\sharp\left(\omega_1+\omega_2\right)\right)\left[\left(\mbox{Ad}_{g_B^{-1}}^\sharp\left(\omega_1+\omega_2\right)\right)^\flat\right]\cr
&=\frac{1}{2}\left(\mbox{Ad}_{g_B^{-1}}^\sharp\left(\omega_1+\omega_2\right)\right)\left[\mbox{Ad}_{g_B^{-1}}\left(\omega_1+\omega_2\right)^\flat\right]\cr
&=\frac{1}{2}\omega_1\left(\omega_1^\flat\right)+\frac{1}{2}\omega_2\left(\omega_2^\flat\right)+\omega_1\left(\omega_2^\flat\right).
\end{align*}
Thus
\begin{multline*}
\left.\dif H_{\mu\nu}\right|_{\left(\omega_1,\omega_2\right)}\left(\pi_{\mathfrak{b}^0}\left(\mbox{ad}^\sharp_\xi\omega_1\right),\pi_{\mathfrak{a}^0}\left(\mbox{ad}^\sharp_\zeta\omega_2\right)\right)=\cr
=\pi_{\mathfrak{b}^0}\left(\mbox{ad}^\sharp_\xi\omega_1\right)\left(\omega_1^\flat+\omega_2^\flat\right)+\pi_{\mathfrak{a}^0}\left(\mbox{ad}^\sharp_\zeta\omega_2\right)\left(\omega_1^\flat+\omega_2^\flat\right)
\end{multline*}
and the hamiltonian vector field will be
\begin{equation}\label{CampoVectHmunu}
\left.X_{H_{\mu\nu}}\right|_{\left(\omega_1,\omega_2\right)}=\left(\pi_{\mathfrak{b}^0}\left(\mbox{ad}^\sharp_{\pi_{\mathfrak{a}}\left(\omega_1^\flat+\omega_2^\flat\right)}\omega_1\right),\pi_{\mathfrak{a}^0}\left(\mbox{ad}^\sharp_{\pi_{\mathfrak{b}}\left(\omega_1^\flat+\omega_2^\flat\right)}\omega_2\right)\right).
\end{equation}
Generally the solving of dynamical systems obtained through reduction proceed in the opposite direction to the used here: That is, once the reduced system is obtained, it is hoped that it can be more easy to solved than the original, because it has less degrees of freedom. After that, the solution for the original problem is obtained by a lifting. In AKS systems the direction in which we solved the problem is reversed: Let us note that the dynamical system on $G\times\mathfrak{g}^*$ defined through $H$ has as hamiltonian vector field
\[
\left.X_H\right|_{\left(g,\sigma\right)}=\left(\sigma^\flat,-\mbox{ad}^\sharp_{\sigma^\flat}\sigma\right)=\left(\sigma^\flat,0\right)
\]
given the condition $\mbox{ad}^\sharp_\xi\xi^\flat=0$ for all $\xi\in\mathfrak{g}$. Therefore the solution curve in the unreduced space passing through $\left(g,\sigma\right)$ at $t=0$ is
\[
t\mapsto\left(g\exp{t\sigma^\flat},\sigma\right),
\]
and if this initial data verifies $\pi_{\mathfrak{b}^0}\left(\mbox{Ad}^\sharp_g\sigma\right)=\mu,\pi_{\mathfrak{a}^0}\left(\sigma\right)=\nu$, then such curve remains in $\Lambda_{\mu\nu}$ for all $t$. Then the curve
\[
t\mapsto\left(\pi_{\mathfrak{b}^0}\left(\mbox{Ad}^\sharp_{\left(g_A\left(t\right)\right)^{-1}}\mu\right),\pi_{\mathfrak{a}^0}\left(\mbox{Ad}^\sharp_{g_B\left(t\right)}\nu\right)\right)
\]
(where $g_A:\mathbb{R}\rightarrow A,g_B:\mathbb{R}\rightarrow B$ are the curves uniquely defined by the equation $g_A\left(t\right)g_B\left(t\right)=g\exp{t\sigma^\flat}$) is a solution curve for the ODE associated to the vector field \eqref{CampoVectHmunu}, which is more difficult to solve than the original system.

\bibliographystyle{plain}

\end{document}